\documentclass[12pt,reqno]{article}
\usepackage{amssymb,amscd,amsmath,amsthm,color}
\newtheorem{thm}{Theorem}[section]

\newtheorem{lemma}[thm]{Lemma}
\newtheorem{cor}[thm]{Corollary}

\newtheorem{remark}[thm]{Remark}
\newtheorem{example}[thm]{Example}

\numberwithin{equation}{section}

\def\bM{\mathbb{M}}
\def\bP{\mathbb{P}}
\def\bH{\mathbb{H}}
\def\cD{\mathcal{D}}
\def\bN{\mathbb{N}}
\def\bR{\mathbb{R}}
\def\bC{\mathbb{C}}
\def\HS{\mathrm{HS}}
\def\WYD{\mathrm{WYD}}
\def\<{\langle}
\def\>{\rangle}
\def\Tr{\mathrm{Tr}\,}
\def\Im{\mathrm{Im}\,}
\def\cH{\mathcal{H}}
\def\ffi{\varphi}
\def\eps{\varepsilon}
\def\cA{\mathcal{A}}
\def\cU{\mathcal{U}}
\def\pTr{\mathrm{Tr}_2\,}
\def\sym{\mathrm{sym}}
\def\cI{\mathcal{I}}

\begin{document}
\baselineskip=16pt
\allowdisplaybreaks

\ \vskip 1cm 
\centerline{\LARGE Convexity of quasi-entropy type functions:}
\medskip
\centerline{\LARGE Lieb's and Ando's convexity theorems revisited}
\bigskip
\bigskip
\centerline{\Large
Fumio Hiai\footnote{E-mail: hiai.fumio@gmail.com}
and D\'enes Petz\footnote{E-mail: petz@math.bme.hu}}

\medskip
\begin{center}
$^1$\,Tohoku University (Emeritus), \\
Hakusan 3-8-16-303, Abiko 270-1154, Japan
\end{center}
\begin{center}
$^2$\,Alfr\'ed R\'enyi Institute of Mathematics, \\ H-1364 Budapest,
POB 127, Hungary
\end{center}

\bigskip
\begin{abstract}
Given a positive function $f$ on $(0,\infty)$ and a non-zero real parameter $\theta$, we
consider a function $I_f^\theta(A,B,X)=\Tr X^*(f(L_AR_B^{-1})R_B)^\theta(X)$ in three
matrices $A,B>0$ and $X$. In the literature $\theta=\pm1$ has been typical. The concept
unifies various quantum information quantities such as quasi-entropy, monotone metrics,
etc. We characterize joint convexity/concavity and monotonicity properties of the function
$I_f^\theta$, thus unifying some known results for various quantum quantities.


\end{abstract}

\section*{Introduction}

The Wigner-Yanase-Dyson (WYD) skew information is an old yet new subject having a somewhat
complicated history since its appearance in the paper \cite{WY} in 1963. The first
fundamental achievements among many things related to the WYD skew information are joint
concavity (also joint convexity) results of Lieb \cite{Li} (the so-called WYDL concavity)
and their equivalent formulations of Ando \cite{An}. The WYDL concavity in the context of
general von Neumann algebras was obtained by Araki \cite{Ar1} in order to prove joint
convexity of the relative entropy, and was further extended by Kosaki \cite{Ko} by means
of interpolation method. The notion of quasi-entropies, extending the relative entropy,
was introduced in \cite{Pe1,Pe2}, and its monotonicity and joint convexity properties
were shown there. A quasi-entropy $S_f^X(\rho\|\sigma)$ for states $\rho,\sigma$ with a
reference operator $X$ is associated with a real function $f$ on $(0,\infty)$, and
operator monotony (or operator convexity) of $f$ is essential in \cite{Pe1,Pe2} as well as
in \cite{Ko}. Moreover, it was proved in \cite{Pe3} that there is a one-to-one
correspondence between the monotone metrics (the quantum version of the Fisher metric in
classical probability) on Riemannian manifolds of positive density matrices and the
(symmetric) operator monotone functions on $(0,\infty)$. A remarkable point here is that
the formula of monotone metrics and that of quasi-entropies are very similar and indeed
they are in dual form (see Subsections 1.2 and 1.3 below). Quasi-entropies and monotone
metrics are among the most important quantities in quantum information and quantum
information geometry. More recent quantum quantities such as the quantum covariance in
\cite{Pe4}, the metric adjusted skew information (generalizing the WYD skew information)
in \cite{Ha2,CH} and the quantum $\chi^2$-divergence in \cite{TKRWV,Ha3} can be
reformulated by quasi-entropies (see Subsection 1.4 and \cite{HP3}).

In the present paper, in the matrix algebra setting (or in finite-dimensional quantum
systems) we deal with the three-variable function
$$
I_f^\theta(A,B,X):=\Tr X^*(f(L_AR_B^{-1})R_B)^\theta(X)
$$
of positive definite matrices $A,B$ and a general matrix $X$ associated with a function
$f>0$ on $(0,\infty)$ and a non-zero real parameter $\theta$, where $L_A$ and $R_B$ are
the left and the right multiplication operators by $A,B$ on matrices. This function unifies
all the quantum quantities mentioned above with particular choices of $\theta$ (typically
$\theta=\pm1$) and of $A,B$ (sometimes $A=B$) as described in Section 1. Indeed, Section 1
may be a concise survey on important quantities in quantum information theory started with
the WYD skew information. In Section 2 we consider various properties concerning joint
convexity as well as joint concavity of the function $I_f^\theta(A,B,X)$ in three variables
$(A,B,X)$ or in two variables $(A,B)$. The main theorem (Theorem \ref{T-2.1}) clarifies
what conditions of $f$ and $\theta$ are sufficient and/or necessary for $I_f^\theta$ to
have those properties. Operator monotony of $f$ shows up and also a possible range of
$\theta$ is rather restricted. The proof is divided into several steps and the results on
operator log-convex/concave functions in \cite{AH} play an essential role in some places.
In Section 3 we consider monotonicity properties of $I_f^\theta$ and show that they are
equivalent to corresponding convexity/concavity properties in Section 2. Furthermore, in
Sections 3 and 4, (joint) convexity properties of the quasi-entropy, the metric adjusted
skew information and the quantum $\chi^2$-divergence are characterized by operator
convexity of the associated function $f$. In this way, we strengthen and unify some known
results on convexity/concavity and monotonicity of several quantum quantities into
characterization (or if and only if) theorems.

\section{Definitions and preliminaries}

For each $n\in\bN$, the $n\times n$ complex algebra is denoted by $\bM_n$, the set of
$n\times n$ Hermitian matrices by $\bH_n$, the set of $n\times n$ positive semidefinite
matrices by $\bM_n^+$, and the set of $n\times n$ positive definite matrices by $\bP_n$.
The usual trace on $\bM_n$ is denoted by $\Tr$. A density matrix is a matrix
$\rho\in\bM_n^+$ with $\Tr\rho=1$. We write $\cD_n$ for the set of $n\times n$ positive
definite density matrices, i.e., $\cD_n:=\{\rho\in\bP_n:\Tr\rho=1\}$. We always consider
$\bM_n$ as a Hilbert space with the {\it Hilbert-Schmidt inner product}
$\<X,Y\>_\HS:=\Tr X^*Y$, $X,Y\in\bM_n$. For any $A\in\bM_n^+$ the {\it left} and the
{\it right multiplications} $L_A$ and $R_A$ are defined as $L_AX:=AX$ and $R_AX:=XA$ for
$X\in\bM_n$, which are commuting positive operators on the Hilbert space
$(\bM_n,\<\cdot,\cdot\>_\HS)$.

For any real function $f$ on $(0,\infty)$ and for every $A,B\in\bP_n$ define a linear
operator $J_{A,B}^f$ on $\bM_n$ by $J_{A,B}^f:=f(L_AR_B^{-1})R_B$ via functional calculus;
more explicitly,
$$
J_{A,B}^f(X)=f(L_AR_B^{-1})R_BX
=\sum_{i=1}^k\sum_{j=1}^lf(\alpha_i\beta_j^{-1})\beta_jP_iXQ_j,
\qquad X\in\bM_n,
$$
where $A=\sum_{i=1}^k\alpha_iP_i$ and $B=\sum_{j=1}^l\beta_jQ_j$ are the spectral
decompositions of $A$ and $B$. In particular, $J_{A,A}^f$ is denoted by $J_A^f$ for short.
Throughout the paper, unless otherwise stated, $f$ is assumed to be strictly positive,
i.e., $f(x)>0$ for all $x\in(0,\infty)$. Then it is immediate to see that $J_{A,B}^f$ is
positive and invertible on $(\bM_n,\<\cdot,\cdot\>_\HS)$ for every $A,B\in\bP_n$. For an
arbitrary real number $\theta$ one can define a three-variable function $I_f^\theta(A,B,X)$
on $\bP_n\times\bP_n\times\bM_n$ by
\begin{equation}\label{F-1.1}
I_f^\theta(A,B,X):=\<X,(J_{A,B}^f)^{-\theta}(X)\>_\HS,
\qquad A,B\in\bP_n,\ X\in\bM_n.
\end{equation}
With the spectral decompositions of $A,B$ as above, $I_f^\theta(A,B,X)$ is more explicitly
written as
\begin{equation}\label{F-1.2}
I_f^\theta(A,B,X)=\sum_{i=1}^k\sum_{j=1}^l(f(\alpha_i\beta_j^{-1})\beta_j)^{-\theta}
\Tr X^*P_iXQ_j.
\end{equation}
When $\theta=0$, $I_f^\theta(A,B,X)$ is reduced to the function $\<X,X\>_\HS$ that is
independent of $A,B$, so we shall always assume that $\theta$ is non-zero.

Our aim of the present paper is to clarify or characterize when the three-variable function
$I_f^\theta(A,B,X)$ is jointly convex or concave in three variables $(A,B,X)$ or in two
variables $(A,B)$. Convexity/concavity properties of this function have been considered by
several authors in its special cases from different viewpoints. Important special cases are
briefly surveyed in the rest of this section, which motivated (also justify) our
consideration of the function $I_f^\theta$ with parameter $\theta$.

\subsection{WYD skew information and WYDL concavity}

The famous {\it Wigner-Yanase-Dyson} ({\it WYD}\,) {\it skew information} introduced in
\cite{WY} is
$$
I_\rho^\WYD(p,K):=-{1\over2}\Tr[\rho^p,K]\,[\rho^{1-p},K]
$$
for $\rho\in\cD_n$, $K\in\bH_n$ and $p\in(0,1)$, where $[X,Y]:=XY-YX$, the commutator.
In his celebrated paper \cite{Li} Lieb proved that
\begin{equation}\label{F-1.3}
(A,B)\in\bM_n^+\times\bM_n^+\longmapsto\Tr X^*A^pXB^q
\end{equation}
is jointly concave for any $X\in\bM_n$ when $p,q\ge0$ and $p+q\le1$. He also proved joint
convexity of $\Tr X^*A^pXB^q$ in three variables $(A,B,X)$ when $p,q\le0$ and $p+q\ge-1$,
and that in two variables $(A,B)$ when $-1\le p,q\le0$. Since
$$
I_\rho^\WYD(p,K)=\Tr\rho K^2-\Tr K\rho^pK\rho^{1-p},
$$
Lieb's concavity settled the convexity question of $\rho\mapsto I_\rho^\WYD(p,K)$, so it 
is also called the {\it WYDL concavity}. For power functions $f(x):=x^\alpha$ on
$(0,\infty)$ with $\alpha\in\bR$, one has
$$
I_f^\theta(A,B,X)=\Tr X^*A^{-\alpha\theta}XB^{-(1-\alpha)\theta}.
$$
For any $p,q\in\bR$ with $p+q\ne0$ there are unique $\alpha,\theta\in\bR$ such that
$-\alpha\theta=p$ and $-(1-\alpha)\theta=q$, so the function $I_f^\theta$ covers trace
functions in \eqref{F-1.3}. On the other hand, in \cite{An} Ando deterimined the range of
real parameters $p,q$ for which the map $(A,B)\in\bP_n\times\bP_n\mapsto A^p\otimes B^q$
is jointly concave (respectively, jointly convex) with respect to the positive
semidefiniteness order. As is well known (see \cite[Proof of 4.3.3]{Bh2},
\cite[Remark 2.6]{CL}) that Ando's convexity/concavity is equivalent to Lieb's, that is,
convexity and concavity of $(A,B)\mapsto\Tr X^*A^pXB^q$ are equivalent to those of
$(A,B)\mapsto A^p\otimes B^q$, respectively.

The WYDL concavity was extended by Araki \cite{Ar1} to the general von Neumann algebra
setting to show joint convexity of the relative entropy, and was further extended by
Kosaki \cite{Ko} based on interpolation theory. Indeed, Kosaki \cite{Ko} proved joint
concavity of
\begin{equation}\label{F-1.4}
(\ffi,\psi)\in M_*^+\times M_*^+\longmapsto
\<x\xi_\psi,f(\Delta_{\ffi,\psi})(x\xi_\psi)\>
\end{equation}
for every $x\in M$ and for every operator monotone function $f\ge0$ on $[0,\infty)$, where
$M_*^+$ is the set of normal positive linear functionals on a von Neumann algebra $M$,
$\xi_\psi$ is the vector representative of $\psi$ in the standard representation of $M$,
and $\Delta_{\ffi,\psi}$ is the {\it relative modular operator} for $\ffi,\psi$
(\cite{Ar1,Ar2}). In the matrix algebra setting with $\ffi=\Tr(A\,\cdot)$ and
$\psi=\Tr(B\,\cdot)$ on $M=\bM_n$ where $A,B\in\bP_n$, one has
$\Delta_{\ffi,\psi}=L_AR_B^{-1}$ and the function in \eqref{F-1.4} reduces to
\begin{equation}\label{F-1.5}
I_f^{-1}(A,B,X)=\<XB^{1/2},f(L_AR_B^{-1})XB^{1/2}\>_\HS.
\end{equation}
which is further reduced to $\Tr X^*A^\alpha XB^{1-\alpha}$ when $f(x)=x^\alpha$.

\subsection{Quasi-entropy}

Quasi-entropies introduced in \cite{Pe1,Pe2} are given by \eqref{F-1.5} in matrix algebras
(and by \eqref{F-1.4} in von Neumann algebras). Thus, the {\it quasi-entropy}
$S_f^X(A\|B)$ for $A,B\in\bP_n$ and $X\in\bM_n$ is nothing but $I_f^\theta(A,B,X)$ with
$\theta=-1$ while the assumption $f>0$ is irrelevant in the definition of $S_f^X(A\|B)$.
This quantity is a generalization of the relative entropy
$$
S(A\|B):=\Tr A(\log A-\log B).
$$
Indeed, $S(A\|B)=S_f^X(A\|B)$ when $f(x)=x\log x$ and $X=I_n$, the $n\times n$ identity
matrix. Monotonicity and joint convexity properties of $S_f^X(A,B)$ were proved in
\cite{Pe1,Pe2}. The convexity result in \cite{Pe1} tells that
$(A,B)\in\bP_n\times\bP_n\mapsto S_f^X(A\|B)$ is jointly convex for every $X\in\bM_n$ if
$f$ is an operator convex function on $(0,\infty)$.

\subsection{Monotone metrics}

The set $\bP_n$ is an open subset of $\bH_n$ that is identified with the $n^2$-dimensional
Euclidean space. Hence $\bP_n$ naturally has a smooth Riemannian manifold structure
so that the tangent space at any foot point is identified with $\bH_n$. When $f$ is an
operator monotone function on $(0,\infty)$, the associated {\it monotone metric} on $\bP_n$
is given by
\begin{equation}\label{F-1.6}
\gamma_A^f(H,K):=\<H,(J_A^f)^{-1}(K)\>_\HS,\qquad A\in\bP_n,\ H,K\in\bH_n.
\end{equation}
Such monotone metrics on the manifold $\bP_n$ (or rather restricted on the submanifold
$\cD_n$) were characterized in \cite{Pe3} in terms of monotonicity under stochastic maps,
i.e., completely positive trace-preserving maps. A monotone metric is also called a
{\it quantum Fisher information} since it is a quantum generalization of the classical
Fisher information. Expression \eqref{F-1.6} makes sense for all $X,Y\in\bM_n$ in place of
$H,K\in\bH_n$ so that $\gamma_A^f(X,X)$ is $I_f^\theta(A,A,X)$ with $\theta=1$. We put
the minus sign of $-\theta$ in definition \eqref{F-1.1} to adjust the parameter to the
expression of monotone metrics. In \cite{HP1,HP2} we discussed Riemannian metrics which
are written as $\<H,(J_A^f)^{-\theta}(K)\>$ for $A\in\bP_n$ (foot point) and $H,K\in\bH_n$
(tangent vectors) when $M(x,y):=f(xy^{-1})y$, $x,y>0$, is a symmetric homogeneous mean.
For $K=H$ this metric is written in the form $I_f^\theta(A,A,H)$.

\subsection{Quantum skew information and quantum $\chi^2$-divergence}

It was observed in \cite{PeHa} (also \cite{CH,Be}) that $I_\rho^\mathrm{WYD}(p,K)$,
$0<p<1$, are expressed, apart from a constant factor, in terms of monotone metrics as
$$
I_\rho^\WYD(p,K)=\gamma_\rho^{f_p}(i[\rho,K],i[\rho,K])
=\<i[\rho,K],(J_\rho^{f_p})^{-1}(i[\rho,K])\>_\HS,
$$
where $f_p$ is an operator monotone function on $(0,\infty)$ defined by
$$
f_p(x):=p(1-p){(x-1)^2\over(x^p-1)(x^{1-p}-1)}.
$$
This observation was extended in \cite{CH} to the paremeter range $-1\le p\le2$.
Furthermore, the WYD skew information was recently generalized by Hansen \cite{Ha2} as
follows: Let $f$ be an operator monotone function on $(0,\infty)$ such that $f$ is
symmetric (i.e., $f(x)=f(x^{-1})x$ for all $x>0$) with $f(1)=1$ and is regular in the
sense that $f(0)$ $(:=\lim_{x\searrow0}f(x))>0$.
The {\it metric adjusted skew information} associated with $f$ is then defined to be
\begin{equation}\label{F-1.7}
I_\rho^f(K):={f(0)\over2}\<i[\rho,K],(J_\rho^f)^{-1}(i[\rho,K])\>_\HS,
\qquad\rho\in\cD_n,\ K\in\bH_n,
\end{equation}
which is written as $I_f^\theta(\rho,\rho,i[\rho,K])$ with $\theta=1$ multiplied by a
constant $f(0)/2$. Moreover, when $f(0)=0$ (non-regular), the {\it unbounded}
({\it metric adjusted}\,) {\it skew information} was defined in \cite{CH} by removing the
constant factor $f(0)/2$ in \eqref{F-1.7}. It was proved in \cite{CH,Ha2} that
$I_\rho^f(K)$ and its unbounded version are convex in $\rho$ for any fixed $K\in\bH_n$.

The {\it quantum $\chi^2$-divergence} recently introduced in \cite{TKRWV} is given by
\begin{equation}\label{F-1.8}
\chi_f^2(\rho,\sigma):=\<\rho-\sigma,(J_\sigma^f)^{-1}(\rho-\sigma)\>_\HS,
\qquad\rho,\sigma\in\cD_n,
\end{equation}
associated with an operator monotone function $f>0$ on $(0,\infty)$. Since
$$
\chi_f^2(\rho,\sigma)=\<\rho,(J_\sigma^f)^{-1}(\rho)\>_\HS-1,
$$
one can rewrite
$$
\chi_f^2(\rho,\sigma)=I_f^1(\sigma,\sigma,\rho-\sigma)=\gamma_\sigma^f(\rho,\rho)-1.
$$
Joint convexity of $(\rho,\sigma)\mapsto\chi_f^2(\rho,\sigma)$ was proved in \cite{TKRWV}
in a special case and generalized in \cite{Ha3} to a general operator monotone function
$f$.

\section{Convexity/concavity of $I_f^\theta(A,B,X)$}
\setcounter{equation}{0}

We begin with enumerating convexity and concavity properties of the function $I_f^\theta$
associated with a function $f>0$ on $(0,\infty)$ and a non-zero real number $\theta$:
\begin{itemize}
\item[(i)] $(A,B,X)\in\bP_n\times\bP_n\times\bM_n\mapsto I_f^\theta(A,B,X)$ is jointly
convex for every $n\in\bN$,
\item[(ii)] $(A,B)\in\bP_n\times\bP_n\mapsto\log I_f^\theta(A,B,X)$ is jointly convex for
any fixed $X\in\bM_n$ and for every $n\in\bN$,
\item[(iii)] $\theta>0$, and $(A,B)\in\bP_n\times\bP_n\mapsto I_f^\theta(A,B,X)$ is jointly
convex for any fixed $X\in\bM_n$ and for every $n\in\bN$,
\item[(iv)]  $(A,B)\in\bP_n\times\bP_n\mapsto I_f^\theta(A,B,X)$ is jointly convex for any
fixed $X\in\bM_n$ and for every $n\in\bN$,
\item[(v)] $(A,B)\in\bP_n\times\bP_n\mapsto I_f^{-\theta}(A,B,X)$ is jointly concave for
any fixed $X\in\bM_n$ and for every $n\in\bN$,
\item[(vi)] $(A,B)\in\bP_n\times\bP_n\mapsto\log I_f^{-\theta}(A,B,X)$ is jointly concave
for any fixed $X\in\bM_n$ and for every $n\in\bN$.
\end{itemize}

For each of the above properties we also consider the property reduced to
$A=B=\rho\in\cD_n$, that is,
\begin{itemize}
\item[(i$'$)] $(\rho,X)\in\cD_n\times\bM_n\mapsto I_f^\theta(\rho,\rho,X)$ is jointly
convex for every $n\in\bN$,
\item[(ii$'$)] $\rho\in\cD_n\mapsto\log I_f^\theta(\rho,\rho,X)$ is convex for any fixed
$X\in\bM_n$ and for every $n\in\bN$,
\end{itemize}
and similarly for (iii$'$)-(vi$'$).

When the Riemannian manifold $\cD_n$ is concerned, the tangent space at each $\rho\in\cD_n$
is $\bH_n^0:=\{H\in\bH_n:\Tr H=0\}$, and $I_f^\theta(\rho,\rho,H)$, $H\in\bH_n^0$, is
considered as a Riemannian metric on $\cD_n$ (see Section 1.3). So, when restricted to
$A=B=\rho\in\cD_n$, it is natural to further restrict $X\in\bM_n$ to $X=H\in\bH_n^0$. We
thus consider the following properties as well:
\begin{itemize}
\item[(i$''$)] $(\rho,H)\in\cD_n\times\bH_n^0\mapsto I_f^\theta(\rho,\rho,H)$ is jointly
convex for every $n\in\bN$,
\item[(ii$''$)] $\rho\in\cD_n\mapsto\log I_f^\theta(\rho,\rho,H)$ is convex for any fixed
$H\in\bH_n^0$ and for every $n\in\bN$,
\end{itemize}
and similarly for (iii$''$)-(vi$''$).

Finally, we present the following intrinsic conditions for $f$ and $\theta$:
\begin{itemize}
\item[(vii)] $f$ is operator monotone on $(0,\infty)$ and $\theta\in(0,1]$,
\item[(viii)] $f$ is operator monotone on $(0,\infty)$ and $\theta\in(0,2]$.
\end{itemize}

Define the {\it $(-\theta)$-power symmetrization} of $f$ by
\begin{equation}\label{F-2.1}
f_{-\theta,\sym}(x)
:=\biggl({f(x)^{-\theta}+\tilde f(x)^{-\theta}\over2}\biggr)^{-1/\theta}
\quad\mbox{where}\quad\tilde f(x):=xf(x^{-1}),\quad x>0.
\end{equation}
When $\theta=-1$, this is the usal symmetrization $(f(x)+\tilde f(x))/2$.

The next theorem is our main result in this section. Note that implication
(vii) $\Rightarrow$ (i) was proved in \cite{HP3} and (i) $\Leftrightarrow$ (i$'$) $\Leftrightarrow$ (vii) for fixed $\theta=1$ was also shown there.

\begin{thm}\label{T-2.1}
Concerning the above properties the following implications hold:
\begin{itemize}
\item[\rm(a)] Each of {\rm(i)}--{\rm(vi)} is equivalent to the corresponding condition with
prime.
\item[\rm(b)] Each of {\rm(i)}--{\rm(vi)} for $f_{-\theta,\sym}$ in place of $f$ is
equivalent to the corresponding condition with double prime for $f$. Consequently, if $f$
is symmetric, i.e., $f(x)=f(x^{-1})x$ for all $x>0$, then each of {\rm(i)}--{\rm(vi)} is
equivalent to the corresponding condition with double prime.
\item[\rm(c)] {\rm(vii)} $\Leftrightarrow$ {\rm(i)}
$\Rightarrow$ {\rm(ii)} $\Leftrightarrow$ {\rm(iii)}
$\Rightarrow$ {\rm(viii)}.
\item[\rm(d)] {\rm(vii)} $\Leftrightarrow$ {\rm(v)}
$\Rightarrow$ {\rm(vi)} $\Rightarrow$ {\rm(viii)}.
\item[\rm(e)] {\rm(iii)} $\Rightarrow$ {\rm(iv)}
$\Rightarrow$ $\theta\in[-2,-1]\cup(0,2]$.
\end{itemize}
\end{thm}

\noindent{\it Proof of {\rm(a)}.}\enspace
The proof is an easy application of the $2\times2$ block matrix trick. For each
$A,B\in\bP_n$ set
$$
\tilde A:=\bmatrix A&0\\0&B\endbmatrix\in\bP_{2n}.
$$
For any $\tilde X=\bmatrix X_{11}&X_{12}\\X_{21}&X_{22}\endbmatrix\in\bM_{2n}$
($=\bM_n\otimes\bM_2$), since
$$
L_{\tilde A}\tilde X=\tilde A\tilde X
=\bmatrix AX_{11}&AX_{12}\\BX_{21}&BX_{22}\endbmatrix,\quad
R_{\tilde A}\tilde X=\tilde X\tilde A
=\bmatrix X_{11}A&X_{12}B\\X_{21}A&X_{22}B\endbmatrix,
$$
one can write
$$
L_{\tilde A}=L_A\oplus L_A\oplus L_B\oplus L_B,\quad
R_{\tilde A}=R_A\oplus R_B\oplus R_A\oplus R_B
$$
under the identification of the Hilbert space $(\bM_{2n},\<\cdot,\cdot\>_\HS)$ with the
direct sum $\bM_n\oplus\bM_n\oplus\bM_n\oplus\bM_n$ via the isomorphism
$X\mapsto X_{11}\oplus X_{12}\oplus X_{21}\oplus X_{22}$. We thus have
\begin{equation}\label{F-2.2}
J_{\tilde A}^f=J_A^f\oplus J_{A,B}^f\oplus J_{B,A}^f\oplus J_B^f
\end{equation}
so that
$$
I_f^\theta\biggl(\tilde A,\tilde A,\bmatrix0&X\\0&0\endbmatrix\biggr)
=\biggl\<\bmatrix0&X\\0&0\endbmatrix,
(J_{\tilde A}^f)^{-\theta}\bmatrix0&X\\0&0\endbmatrix\biggr\>_\HS
=I_f^\theta(A,B,X),
$$
from which each of (i)--(vi) is equivalent to the corresponding condition reduced to
$A=B\in\bP_n$. It remains to show that the latter condition is also equivalent to the
condition further reduced to $A=B\in\cD_n$. Since
$I_f^\theta(cA,cA,X)=c^{-\theta}I_f(A,A,X)$ for $A\in\bP_n$, $X\in\bM_n$ and $c>0$, the
condition in question is equivalent to that with restriction $\Tr A<1$. For such $A\in\bP_n$ and $X\in\bM_n$ one has
$I_f^\theta(A,A,X)=I_f^\theta(\hat\rho,\hat\rho,\hat X)$ by letting
$\hat\rho:=A\oplus(1-\Tr A)\in\cD_{n+1}$ and $\hat X:=X\oplus0\in\bM_{n+1}$. This
immediately implies the conclusion.\qed

\bigskip
\noindent{\it Proof of {\rm(b)}.}\enspace
Set $g:=f_{-\theta,\sym}$ for brevity. Since
$g^{-\theta}=(f^{-\theta}+\tilde f^{-\theta})/2$, it is obvious that
\begin{equation}\label{F-2.3}
(J_{A,B}^g)^{-\theta}={(J_{A,B}^f)^{-\theta}+(J_{A,B}^{\tilde f})^{-\theta}\over2},
\qquad A,B\in\bP_n.
\end{equation}
Moreover, taking the spectral decompositions $A=\sum_{i=1}^k\alpha_iP_i$ and
$B=\sum_{j=1}^l\beta_jQ_j$, for every $X\in\bM_n$ we have by \eqref{F-1.2}
\begin{align}
I_f^\theta(A,B,X)
&=\sum_{i=1}^k\sum_{j=1}^l(f(\alpha_i\beta_j^{-1})\beta_j)^{-\theta}
\Tr X^*P_iXQ_j \nonumber\\
&=\sum_{i=1}^k\sum_{j=1}^l(\tilde f(\beta_j\alpha_i^{-1}\alpha_i)^{-\theta}
\Tr XQ_jX^*P_i=I_{\tilde f}^\theta(B,A,X^*). \label{F-2.4}
\end{align}
In particular, when $A=B$ and $X=H\in\bH_n$, we have
\begin{equation}\label{F-2.5}
I_f^\theta(A,A,H)=I_{\tilde f}^\theta(A,A,H)=I_g^\theta(A,A,H)
\end{equation}
thanks to \eqref{F-2.3}. For every $\rho\in\cD_n$ and $X\in\bM_n$ set
$$
\tilde \rho:={1\over2}\begin{bmatrix}\rho&0\\0&\rho\end{bmatrix}\in\cD_{2n},\quad
\tilde H:=\begin{bmatrix}0&X\\X^*&0\end{bmatrix}\in\bH_{2n}^0.
$$
Since $(J_{\rho/2}^g)^{-\theta}=2^\theta(J_\rho^g)^{-\theta}$ and $g$ is symmetric (i.e.,
$\tilde g=g$), it follows from \eqref{F-2.2} and \eqref{F-2.4} (for $g$) that
$$
I_g^\theta(\tilde\rho,\tilde\rho,\tilde H)
=2^\theta\{I_g^\theta(\rho,\rho,X)+I_g^\theta(\rho,\rho,X^*)\}
=2^{\theta+1}I_g^\theta(\rho,\rho,X).
$$
Furthermore, by \eqref{F-2.3} and \eqref{F-2.4} we have
$$
I_g^\theta(\tilde\rho,\tilde\rho,\tilde H)
={I_f^\theta(\tilde\rho,\tilde\rho,\tilde H)
+I_{\tilde f}^\theta(\tilde\rho,\tilde\rho,\tilde H)\over2}
=I_f^\theta(\tilde\rho,\tilde\rho,\tilde H).
$$
Therefore,
\begin{equation}\label{F-2.6}
I_f^\theta(\tilde\rho,\tilde\rho,\tilde H)=2^{\theta+1}I_g^\theta(\rho,\rho,X).
\end{equation}
From \eqref{F-2.5} and \eqref{F-2.6} together with (a) one can see that each of (i)--(vi)
for $g$ is equivalent to the corresponding with double prime for $f$. The latter assertion
of (b) is immediate since $g=f$ for symmetric $f$.\qed

\bigskip
The part (c)--(e) is the main assertion of the theorem. The proof is based on
\cite[Theorems 3.1 and 3.7]{AH}, so we first state necessary parts from them as a lemma
for the convenience of the reader. Let $\cH$ be an infinite-dimensional separable
Hilbert space with inner product $\<\cdot,\cdot\>$, and $B(\cH)^{++}$ be the set of all
positive and invertible bounded operators on $\cH$. Let $f>0$ be a continuous function
on $(0,\infty)$, and $f(A)$ be defined for $A\in B(\cH)^{++}$ via functional calculus as
usual. (A function $f>0$ on $(0,\infty)$ is inevitably continuous if it satisfies any of
the conditions listed before Theorem 2.1, so the continuity assumption for $f$ here is
harmless.)

\begin{lemma}\label{L-2.2}
In the above situation the following conditions {\rm(a1)}--{\rm(a4)} are equivalent:
\begin{itemize}
\item[\rm(a1)] $f$ is operator monotone decreasing on $(0,\infty)$;
\item[\rm(a2)] $(A,\xi)\in B(\cH)^{++}\times\cH\mapsto\<\xi,f(A)\xi\>$ is jointly convex;
\item[\rm(a3)] $A\in B(\cH)^{++}\mapsto\log\<\xi,f(A)\xi\>$ is convex for every
$\xi\in\cH$;
\item[\rm(a4)] $f$ is operator convex on $(0,\infty)$ and the numerical function $f(x)$
is non-increasing on $(0,\infty)$.
\end{itemize}

Also, the following conditions {\rm(b1)} and {\rm(b2)} are equivalent:
\begin{itemize}
\item[\rm(b1)] $f$ is operator monotone {\rm(}or equivalently, operator concave{\rm)} on
$(0,\infty)$;
\item[\rm(b2)] $A\in B(\cH)^{++}\mapsto\log\<\xi,f(A)\xi\>$ is concave for every
$\xi\in\cH$.
\end{itemize}
\end{lemma}

Note that log-convexity is stronger than convexity for positive functions while
log-concavity is weaker than concavity. The log-convexity condition (a3) characterizes
operator monotone decreasingness of $f$ that is a stronger version of operator convexity.
On the other hand, the log-concavity condition (b2) is equivalent to operator concavity
of $f$. It is well known \cite[V.2.5]{Bh} that operator monotony and operator concavity
are equivalent for a continuous non-negative function on $(0,\infty)$.

To make the proof of the theorem more tractable, we next present a few more lemmas that
are some technical ingredients of the proof of the part (c)--(e).

\begin{lemma}\label{L-2.3}
Let $f$ be as above, and assume that $(A,\xi)\in\bP_n\times\bC^n\mapsto\<\xi,f(A)\xi\>$ is
jointly convex for every $n\in\bN$, where $\<\cdot,\cdot\>$ is the usual inner product on
$\bC^n$. Then $(A,\xi)\in B(\cH)^{++}\times\cH\mapsto\<\xi,f(A)\xi\>$ is jointly convex.
\end{lemma}

\begin{proof}
The proof is standard by using a familiar convergence argument. Let $\{e_i\}_{i=1}^\infty$
be an orthonormal basis of $\cH$ (in Lemma \ref{L-2.2}). For each $n\in\bN$ let $P_n$ be
the orthogonal projection onto the linear span of $\{e_1,\dots,e_n\}$, and $I$ be the
identity operator on $\cH$. By assumption we see that
$$
(A,\xi)\in B(\cH)^{++}\times\cH\mapsto\<P_n\xi,P_nf(P_nAP_n)P_n\xi\>
$$
is jointly convex, where $f(P_nAP_n)$ is the functional calculus as an operator on
$P_n\cH$ ($\cong\bC^n$). Since
$$
P_nf(P_nAP_n)P_n=P_nf(P_nAP_n+(I-P_n))P_n
$$
converges to $f(A)$ in the strong operator topology, it follows that
$\<P_n\xi,P_nf(P_nAP_n)P_n\xi\>$ converges to $\<\xi,f(A)\xi\>$ as $n\to\infty$. Hence the
conclusion follows.
\end{proof}

\begin{lemma}\label{L-2.4}
Let $f$ be as above. Assume that both $f(x)^\theta$ and $(f(x^{-1})x)^\theta$ are operator
monotone on $(0,\infty)$ for some $\theta\in\bR\setminus\{0\}$. Then $0<\theta\le2$ and
$f$ is operator monotone on $(0,\infty)$, i.e., condition {\rm(viii)} holds.
\end{lemma}

\begin{proof}
Let $g(x):=f(x)^\theta$ and so $(f(x^{-1})x)^\theta=g(x^{-1})x^\theta$ for $x>0$. Since
$g(x)$ and $g(x^{-1})x^\theta$ are operator monotone on $(0,\infty)$ and so they are
non-decreasing and concave functions on $(0,\infty)$, there are $\delta>0$ and $a,b,c,d>0$
such that $ax\le g(x)\le b$ for all $x\in(0,\delta)$ and $c\le g(x^{-1})x^\theta \le dx$
for all $x\in(\delta^{-1},\infty)$. The latter restriction yields that
$cx^\theta\le g(x)\le dx^{\theta-1}$ for all $x\in(0,\delta)$. Combining this with the
former, we have $cx^\theta\le b$ and $ax\le dx^{\theta-1}$ for $x\in(0,\delta)$, which
implies that $\theta\ge0$ and $1\ge\theta-1$ so that $0<\theta\le2$ since $\theta\ne0$ by
assumption.

To prove the operator monotony of $f$, we may assume that $f$ is not a constant function.
Due to L\"owner's theorem (see \cite[V.4.7]{Bh}, \cite[Theorem 2.7.7]{Hi}), the
functions $g(x)$ on $(0,\infty)$ is analytically continued to a Pick function $g(z)$ on
$\bC^+\cup\bC^-$, where $\bC^+\,(\bC^-):=\{z\in\bC:\Im z>0\,(<0)\}$, so that
$g(\bC^+)\subset\bC^+$ and $g(\bC^-)\subset\bC^-$. Then $g(x^{-1})x^\theta$ can be
analytically continued to $g(z^{-1})z^\theta$ for $z\in\bC^+\cup\bC^-$, where $z^\theta$
is defined with the usual branch. By assumption, $g(z^{-1})z^\theta$ must be a Pick
function again. Now let $z=r^{-1}e^{i\eta}\in\bC^+$ with $r>0$ and $0<\eta<\pi$. Since
$$
g(re^{-i\eta})\cdot r^{-\theta}e^{i\theta\eta}=g(z^{-1})z^\theta\in\bC^+
$$
and
$$
g(re^{-i\eta})=g(\overline{r^2z})=\overline{g(r^2z)}=\overline{g(re^{i\eta})},
$$
we have $\overline{g(re^{i\eta})}e^{i\theta\eta}\in\bC^+$. Noting that the argument of
$\overline{g(re^{i\eta})}e^{i\theta\eta}$ changes continuously as $\eta$ changes in
$(0,\pi)$, we obtain $\theta\eta-\arg g(re^{i\eta})\ge0$ so that, thanks to $\theta>0$,
$$
\arg g(re^{i\eta})^{1/\theta}={1\over\theta}\arg g(re^{i\eta})\le\eta
$$
for all $\eta\in(0,\pi)$. This shows that $f(x)=g(x)^{1/\theta}$ is analytically continued
to a Pick function $g(z)^{1/\theta}$. Hence $f$ is operator monotone by L\"owner's theorem.
\end{proof}

\begin{lemma}\label{L-2.5}
Let $g$ be an operator convex function on $(0,\infty)$ such that $g(x)\ge0$ for all $x>0$
and $g$ is not identically zero. When $g(+0):=\lim_{x\searrow0}g(x)=0$, there are
$\delta>0$ and $a,b,c,d>0$ such that
$$
\begin{cases}
ax^2\le g(x)\le bx & \text{if $0<x<\delta$}, \\
cx\le g(x)\le dx^2 & \text{if $\delta^{-1}<x<\infty$}.
\end{cases}
$$
When $g(+0)\in(0,\infty]$, there are $\delta>0$ and $a,b,c,d>0$ such that
$$
\begin{cases}
a\le g(x)\le bx^{-1} & \text{if $0<x<\delta$}, \\
cx^{-1}\le g(x)\le dx^2 & \text{if $\delta^{-1}<x<\infty$}.
\end{cases}
$$
\end{lemma}

\begin{proof}
Assume that $g(+0)=0$. Then $g'(+0):=\lim_{x\searrow0}g(x)/x$ exists in $[0,\infty)$. Hence
$g$ admits the integral expression
\begin{equation}\label{F-2.7}
g(x)=\beta x+\gamma x^2+\int_{(0,\infty)}{x^2(1+\lambda)\over x+\lambda}\,d\mu(\lambda),
\qquad x>0,
\end{equation}
where $\beta=g'(+0)\ge0$, $\gamma\ge0$ and $\mu$ is a finite positive measure on
$(0,\infty)$ (see \cite[V.5.5]{Bh}). By the monotone convergence and the Lebesgue
convergence theorems, we notice that
\begin{align*}
\lim_{x\searrow0}\int_{(0,\infty)}{1+\lambda\over x+\lambda}\,d\mu(\lambda)
&=\int_{(0,\infty)}{1+\lambda\over\lambda}\,d\mu(\lambda), \\
\lim_{x\searrow0}\int_{(0,\infty)}{x(1+\lambda)\over x+\lambda}\,d\mu(\lambda)&=0, \\
\lim_{x\to\infty}\int_{(0,\infty)}{x(1+\lambda)\over x+\lambda}\,d\mu(\lambda)
&=\int_{(0,\infty)}(1+\lambda)\,d\mu(\lambda), \\
\lim_{x\to\infty}\int_{(0,\infty)}{1+\lambda\over x+\lambda}\,d\mu(\lambda)&=0,
\end{align*}
which yield that
\begin{align}
\lim_{x\searrow0}{g(x)\over x^2}
&=\beta\cdot(+\infty)+\gamma+\int_{(0,\infty)}{1+\lambda\over\lambda}\,d\mu(\lambda),
\label{F-2.8}\\
\lim_{x\searrow0}{g(x)\over x}&=\beta, \nonumber\\
\lim_{x\to\infty}{g(x)\over x}
&=\beta+\gamma\cdot(+\infty)+\int_{(0,\infty)}(1+\lambda)\,d\mu(\lambda),
\label{F-2.9}\\
\lim_{x\to\infty}{g(x)\over x^2}&=\gamma. \nonumber
\end{align}
Note that each of \eqref{F-2.8} and \eqref{F-2.9} is strictly positive; otherwise,
$\beta=\gamma=0$ and $\mu$ is a zero measure so that $g$ is identically zero, contradicting
assumption. Now, the first assertion of the lemma follows immediately.

Next, assume that $g(+0)\in(0,\infty]$. We use other types of integral representations for
operator convex functions. The function $g(x+1)$ restricted on $(-1,1)$ admits the
expression
$$
g(x+1)=\alpha+\beta x+\int_{[-1,1]}{x^2\over1-\lambda x}\,d\mu(\lambda),
\qquad x\in(-1,1),
$$
with $\alpha,\beta\in\bR$ and $\mu$ is a finite positive measure on $[-1,1]$
(see \cite[V.4.6]{Bh}). So we write
$$
g(x)=\alpha+\beta(x-1)+\int_{[-1,1]}{(x-1)^2\over1+\lambda(1-x)}\,d\mu(\lambda),
\qquad x\in(0,2).
$$
Therefore,
$$
\lim_{x\searrow0}xg(x)=\mu(\{-1\}).
$$
which implies that $a\le g(x)\le bx^{-1}$ for some $a,b>0$ and for all sufficiently small
$x>0$. Finally, we examine the order of $g$ as $x\to\infty$. Assume that $g$ is
non-increasing on $(0,\infty)$. By (a4) $\Rightarrow$ (a1) of Lemma \ref{L-2.2}, $g$ is
operator monotone decreasing on $(0,\infty)$. Hence, as shown in \cite{Ha1} (also see the
proof of \cite[Theorem 3.1]{AH}), we have the expression
$$
g(x)=\alpha+\int_{[0,\infty)}{1+\lambda\over x+\lambda}\,d\mu(\lambda),
\qquad x>0,
$$
where $\alpha\ge0$ and $\mu$ is a finite positive measure on $[0,\infty)$. Since
$$
\lim_{x\to\infty}xg(x)=\alpha\cdot(+\infty)+\int_{[0,\infty)}(1+\lambda)\,d\mu(\lambda)
$$
is strictly positive, we have $cx^{-1}\le g(x)\le d$ for some $c,d>0$ and for all
sufficiently large $x>0$. Next, assume that $g$ is not entirely non-increasing on
$(0,\infty)$, so there is a $\kappa\in(0,\infty)$ such that $g'(\kappa)>0$. Then
$g(x+\kappa)-g(\kappa)$ on $(0,\infty)$ admits the same expression as \eqref{F-2.7}. Now,
as in the proof when $g(+0)=0$, one can show that $cx\le g(x)\le dx^2$ for some $c,d>0$
and for all $x>0$ large enough. Hence the second assertion of the lemma has been shown.
\end{proof}

\noindent
{\it Proof of {\rm(c)}.}\enspace
(vii) $\Rightarrow$ (i) was proved in \cite[Theorem 7]{HP3} based on joint concavity of
operator means \cite{KA} (note that $J_{A,B}^f$ is the operator obtained by applying the
operator mean associated with an operator monotone function $f$ to $R_B$ and $L_A$).

%Assume that $f$ is operator monotone on $(0,\infty)$. We may further assume normalization
%$f(1)=1$ without loss of generality. Then
%$$
%J_{A,B}^f=R_B^{1/2}f(R_B^{-1/2}L_AR_B^{-1/2})R_B^{1/2}
%$$
%is the operator obtained by applying the operator mean associated with $f$ to $R_B$ and
%$L_A$ (see \cite{KA}). Thanks to joint concavity of operator means
%\cite[Theorem 3.5]{KA}, for every $A_1,A_2,B_1,B_2\in\bP_n$ we have
%\begin{equation}\label{F-2.10}
%J_{{A_1+A_2\over2},{B_1+B_2\over2}}^f\ge{J_{A_1,B_1}^f+J_{A_2,B_2}^f\over2}.
%\end{equation}
%Since $x^{-\theta}$ on $(0,\infty)$ is operator monotone decreasing thanks to
%$0<\theta\le1$, it follows from (a1) $\Rightarrow$ (a2) of Lemma \ref{L-2.2} that for
%every $X_1,X_2\in\bM_n$,
%\begin{align*}
%&I_f^\theta\biggl({A_1+A_2\over2},{B_1+B_2\over2},{X_1+X_2\over2}\biggr) \\
%&\qquad\le\Biggl\<{X_1+X_2\over2},
%\Biggl({J_{A_1,B_1}^f+J_{A_2,B_2}^f\over2}\Biggr)^{-\theta}
%\biggl({X_1+X_2\over2}\biggr)\Biggr\>_\HS \\
%&\qquad\le{1\over2}\bigl\{\<X_1,(J_{A_1,B_1}^f)^{-\theta}X_1\>_\HS
%+\<X_2,(J_{A_2,B_2}^f)^{-\theta}X_2\>_\HS\bigr\} \\
%&\qquad={1\over2}\bigl\{I_f^\theta(A_1,B_1,X_1)+I_f^\theta(A_2,B_2,X_2)\bigr\}.
%\end{align*}
%Hence (i) follows.

\medskip
(i) $\Rightarrow$ (vii).\enspace
Let $n\in\bN$ be arbitrary. For each $\xi\in\bC^n$ let $X_\xi:=[\xi\,0\,\cdots\,0]\in\bM_n$,
i.e., the first column of $X_\xi$ is $\xi$ and all other entries of $X$ are zero. When
$B=I_n$ and $X=X_\xi$, we have
\begin{equation}\label{F-2.10}
I_f^\theta(A,I_n,X_\xi)=\<X_\xi,f(L_A)^{-\theta}(X_\xi)\>_\HS
=\<X_\xi,f(A)^{-\theta}X_\xi\>_\HS=\<\xi,f(A)^{-\theta}\xi\>.
\end{equation}
Moreover, when $A=I_n$ and $X=X_\xi^t$, the transpose of $X_\xi$, we have
\begin{align}
I_f^\theta(I_n,B^t,X_\xi^t)
&=\<X_\xi^t,(f(R_{B^t}^{-1})R_{B^t})^{-\theta}(X_\xi^t)\>_\HS
=\<X_\xi^t,X_\xi^t((f(B^{-1})B)^{-\theta})^t\>_\HS \nonumber\\
&=\<X_\xi^t,((f(B^{-1})B)^{-\theta}X_\xi)^t\>_\HS
=\<X_\xi,(f(B^{-1})B)^{-\theta}X_\xi\>_\HS \nonumber\\
&=\<\xi,(f(B^{-1})B)^{-\theta}\xi\>.
\label{F-2.11}
\end{align}
Hence (i) implies that $\<\xi,f(A)^{-\theta}\xi\>$ and $\<\xi,(f(A^{-1})A)^{-\theta}\xi\>$
are jointly convex in $(A,\xi)\in\bP_n\times\bC^n$ for every $n\in\bN$, so by Lemma
\ref{L-2.3} they are jointly convex in $(A,\xi)\in B(\cH)^{++}\times\cH$ in the situation
of Lemma \ref{L-2.2}. Then by (a2) $\Rightarrow$ (a1) of Lemma \ref{L-2.2}, both
$f(x)^{-\theta}$ and $(f(x^{-1})x)^{-\theta}$ are operator monotone decreasing on
$(0,\infty)$, so both $f(x)^\theta$ and $(f(x^{-1})x)^\theta$ are operator monotone on
$(0,\infty)$. Hence Lemma \ref{L-2.4} implies that (viii) holds.

Now, it remains to prove that $\theta\le1$. To do so, define the function
$$
\ffi(x,y,z):=n^{-1}I_f^\theta(xI_n,yI_n,zI_n)
=\phi(x,y)^{-\theta}z^2,\qquad x,y,z>0,
$$
where $\phi(x,y):=f(xy^{-1})y$, and compute the Hessian of $\ffi$ as follows:
\begin{align}
&{\small
\det\bmatrix z^2\{\theta(\theta+1)\phi^{-\theta-2}\phi_x^2
-\theta\phi^{-\theta-1}\phi_{xx}\}&
z^2\{\theta(\theta+1)\phi^{-\theta-2}\phi_x\phi_y
-\theta\phi^{-\theta-1}\phi_{xy}\}&
-2\theta z\phi^{-\theta-1}\phi_x \\
z^2\{\theta(\theta+1)\phi^{-\theta-2}\phi_x\phi_y
-\theta\phi^{-\theta-1}\phi_{xy}\}&
z^2\{\theta(\theta+1)\phi^{-\theta-2}\phi_y^2
-\theta\phi^{-\theta-1}\phi_{yy}\}&
-2\theta z\phi^{-\theta-1}\phi_y \\
-2\theta z\phi^{-\theta-1}\phi_x&
-2\theta z\phi^{-\theta-1}\phi_y&
2\phi^{-\theta}\endbmatrix} \nonumber\\
&\qquad=2\theta^2z^4\phi^{-3\theta-4}
\det\bmatrix(\theta+1)\phi_x^2-\phi\phi_{xx}&
(\theta+1)\phi_x\phi_y-\phi\phi_{xy}&-2\theta \phi_x \\
(\theta+1)\phi_x\phi_y-\phi\phi_{xy}&
(\theta+1)\phi_y^2-\phi\phi_{yy}&-2\theta\phi_y \\
-\phi_x&-\phi_y&1\endbmatrix \nonumber\\
&\qquad=2\theta^2z^4\phi^{-3\theta-4}
\det\bmatrix-\phi\phi_{xx}&-\phi\phi_{xy}&-(\theta-1)\phi_x \\
-\phi\phi_{xy}&-\phi\phi_{yy}&-(\theta-1)\phi_y \\
-\phi_x&-\phi_y&1\endbmatrix \nonumber\\
&\qquad=2\theta^2z^4\phi^{-3\theta-3}
\det\bmatrix \phi_{xx}&\phi_{xy}&-(\theta-1)\phi_x \\
\phi_{xy}&\phi_{yy}&-(\theta-1)\phi_y \\
\phi_x&\phi_y&\phi\endbmatrix \nonumber\\
&\qquad=2\theta^2z^4\phi^{-3\theta-3}
\bigl\{(\theta-1)(\phi_x^2\phi_{yy}+\phi_y^2\phi_{xx}-2\phi_x\phi_y\phi_{xy})
+\phi(\phi_{xx}\phi_{yy}-\phi_{xy}^2)\bigr\}. \label{F-2.12}
\end{align}
We further compute
\begin{equation}\label{F-2.13}
\begin{cases}
\phi_x(x,y)=f'(xy^{-1}), \\
\phi_y(x,y)=f(xy^{-1})-f'(xy^{-1})(xy^{-1}), \\
\phi_{xx}(x,y)=f''(xy^{-1})y^{-1}, \\
\phi_{xy}(x,y)=-f''(xy^{-1})xy^{-2}, \\
\phi_{yy}(x,y)=f''(xy^{-1})(x^2y^{-3}).
\end{cases}
\end{equation}
Insert these formulas when $y=1$ into \eqref{F-2.12} to obtain the Hessian of $\ffi$ at
$(x,1,1)$ as
$$
2\theta^2(\theta-1)f(x)^{-3\theta-1}f''(x),
$$
which should be non-negative for any $x>0$. Suppose that $\theta>1$; then it must follow
that $f''(x)\ge0$, so $f$ is convex. Moreover, $f(x)^\theta$ is concave since it is
operator monotone (hence operator concave). Hence $f$ must be a constant function so that
$(f(x^{-1})x)^\theta$ is $x^\theta$ up to a multiple constant. This means that $x^\theta$
is operator monotone, which contradicts $\theta>1$. So we have $\theta\le1$.

\medskip
(i) $\Rightarrow$ (iii) is obvious since (i) implies (vii) as shown above.

\medskip
(ii) $\Rightarrow$ (iii).\enspace
Since the function
$$
\log I_f^\theta(xI_n,xI_n,I_n)=-\theta\log x-\theta\log f(1)+\log n
$$
is convex in $x>0$, we have $\theta>0$. Hence (iii) follows because a positive function is
convex if its logarithm is convex.

\medskip
(iii) $\Rightarrow$ (ii).\enspace
This can be proved in the same method adopted in \cite{Li} while we sketch the proof for
the convenience of the reader. Let $A_1,A_2,B_1,B_2\in\bP_n$ and $X\in\bM_n$ be arbitrary,
and define
$$
\phi(x):=I_f^\theta(x_1A_1+x_2A_2,x_1B_1+x_2B_2,X)
$$
for $x=(x_1,x_2)\in Q:=[0,\infty)^2\setminus\{(0,0)\}$. The function $\phi$ on $Q$ is
convex by (iii), and we need to prove that $\log\phi$ is convex. To do so, we may assume
that $X\ne0$ and hence $\phi(x)>0$ for all $x\in Q$. Since
$\log\phi(x)=\lim_{r\searrow0}(\phi(x)^r-1)/r$, it suffices to prove that $\phi(x)^r$ is
convex on $Q$ for any $r>0$. For each $\alpha>0$ define a convex set
$G_\alpha:=\{x\in Q:\phi(x)\le\alpha\}$; then $\phi(x)=\inf\{\alpha:x\in G_\alpha\}$. Since
$\phi(\lambda x)=\lambda^{-\theta}\phi(x)$ for all $\lambda>0$ as easily checked, we have
$G_\alpha=\alpha^{-1/\theta}G_1$. Let
$$
k(x):=\sup\{\mu>0:x\in G_{\mu^{-\theta}}\}=\sup\{\mu>0:\mu^{-1}x\in G_1\},
\qquad x\in Q.
$$
Then $k(x)$ is positive and concave on $Q$, and moreover we have
$$
k(x)=\sup\{\lambda^{-1/\theta}:x\in G_\lambda,\ \lambda>0\}=\phi(x)^{-1/\theta}.
$$
Therefore, for each $r>0$, $\phi(x)^r=k(x)^{-\theta r}$ is convex since $\theta r>0$.

\medskip
(ii) $\Rightarrow$ (viii).\enspace
Thanks to \eqref{F-2.10} and \eqref{F-2.11} it follows from (ii) that
$\log\<\xi,f(A)^{-\theta}\xi\>$ and $\log\<\xi,(f(A^{-1})A)^{-\theta}\xi\>$ are convex in
$A\in\bP_n$ for any fixed $\xi\in\bC_n$ and for every $n$. Similarly to the proof of Lemma
\ref{L-2.3}, we see that $\log\<\xi,f(A)^{-\theta}\xi\>$ and
$\log\<\xi,(f(A^{-1})A)^{-\theta}\xi\>$ are convex in $A\in B(\cH)^{++}$ for every
$\xi\in\cH$ in the situation of Lemma \ref{L-2.2}. So, (a3) $\Rightarrow$ (a1) of Lemma
\ref{L-2.2} yields that both $f(x)^{-\theta}$ and $(f(x^{-1})x)^{-\theta}$ are operator
monotone decreasing on $(0,\infty)$, that is, both $f(x)^\theta$ and $(f(x^{-1})x)^\theta$
are operator monotone on $(0,\infty)$. Hence Lemma \ref{L-2.4} implies that (viii) holds.
\qed

\bigskip\noindent
{\it Proof of {\rm(d)}.}\enspace
(vii) $\Rightarrow$ (v).\enspace
As mentioned in the proof of \cite[Theorem 7]{HP3} we have joint concavity of
$(A,B)\in\bP_n\times\bP_n\mapsto J_{A,B}^f$, that is, for every $A_1,A_2,B_1,B_2\in\bP_n$,
$$
J_{{A_1+A_2\over2},{B_1+B_2\over2}}^f\ge{J_{A_1,B_1}^f+J_{A_2,B_2}^f\over2}.
$$
Operator monotony and operator concavity of $x^\theta$ on $(0,\infty)$ give
$$
\biggl(J_{{A_1+A_2\over2},{B_1+B_2\over2}}^f\biggr)^\theta
\ge\Biggl({J_{A_1,B_1}^f+J_{A_2,B_2}^f\over2}\Biggr)^\theta
\ge{\bigl(J_{A_1,B_1}^f\bigr)^\theta+\bigl(J_{A_2,B_2}^f\bigr)^\theta\over2},
$$
which implies (v).

\medskip
(v) $\Rightarrow$ (vi) is trivial because the logarithm of a positive concave function is
concave.

\medskip
(vi) $\Rightarrow$ (viii).\enspace
The proof is similar to the above proof of (ii) $\Rightarrow$ (viii) of (c). With $-\theta$
in place of $\theta$ and with ``concave" in place of ``convex", we see from (vi) that
$\log\<\xi,f(A)^\theta\xi\>$ and $\log\<\xi,(f(A^{-1})A)^\theta\xi\>$ are concave in
$A\in B(\cH)^{++}$ for every $\xi\in\cH$. Hence by (b2) $\Rightarrow$ (b1) of Lemma
\ref{L-2.2}, both $f(x)^\theta$ and $(f(x^{-1})x)^\theta$ are operator monotone on
$(0,\infty)$, so Lemma \ref{L-2.4} implies (viii).

\medskip
(v) $\Rightarrow$ (vii).\enspace
Since (v) $\Rightarrow$ (viii) is already known, it remains to prove that $\theta\le1$. But
this is immediately seen because the function
$$
I_f^{-\theta}(xI_n,xI_n,I_n)=nf(1)^\theta x^\theta
$$
is concave in $x>0$.\qed

\bigskip\noindent
{\it Proof of {\rm(e)}.}\enspace
(iii) $\Rightarrow$ (iv) is trivial. Finally, we prove that (iv) implies the restriction
that $\theta\in[-2,-1]\cup(0,2]$. Since the function
$$
I_f^\theta(xI_n,xI_n,I_n)=nf(1)^{-\theta}x^{-\theta}
$$
is convex in $x>0$, we have $\theta\in(-\infty,-1]\cup(0,\infty)$. When $\theta>0$, (iv)
is (iii), so $\theta\in(0,2]$ follows from (iii) $\Rightarrow$ (viii). Now assume that
$\theta<0$. We need to prove that $\theta\ge-2$. Thanks to \eqref{F-2.10} and
\eqref{F-2.11} we see from (iv) that both $f(x)^{-\theta}$ and $(f(x^{-1})x)^{-\theta}$ are
operator convex on $(0,\infty)$. Let $g(x):=f(x)^{-\theta}$ and so
$(f(x^{-1})x)^{-\theta}=g(x^{-1})x^{-\theta}$. By Lemma \ref{L-2.5} one can choose
$\delta>0$ and $a,b,c,d>0$ such that either
\begin{equation}\label{F-2.14}
\begin{cases}
ax^2\le g(x)\le bx & \text{if $0<x<\delta$}, \\
cx\le g(x)\le dx^2 & \text{if $\delta^{-1}<x<\infty$},
\end{cases}
\end{equation}
or
\begin{equation}\label{F-2.15}
\begin{cases}
a\le g(x)\le bx^{-1} & \text{if $0<x<\delta$}, \\
cx^{-1}\le g(x)\le dx^2 & \text{if $\delta^{-1}<x<\infty$},
\end{cases}
\end{equation}
and also either
\begin{equation}\label{F-2.16}
\begin{cases}
ax^2\le g(x^{-1})x^{-\theta}\le bx & \text{if $0<x<\delta$}, \\
cx\le g(x^{-1})x^{-\theta}\le dx^2 & \text{if $\delta^{-1}<x<\infty$},
\end{cases}
\end{equation}
or
\begin{equation}\label{F-2.17}
\begin{cases}
a\le g(x^{-1})x^{-\theta}\le bx^{-1} & \text{if $0<x<\delta$}, \\
cx^{-1}\le g(x^{-1})x^{-\theta}\le dx^2 & \text{if $\delta^{-1}<x<\infty$},
\end{cases}
\end{equation}
Assume that \eqref{F-2.14} and \eqref{F-2.17} are satisfied. Since \eqref{F-2.17} is
rephrased as
$$
\begin{cases}
ax^{-\theta}\le g(x)\le bx^{-\theta+1} & \text{if $\delta^{-1}<x<\infty$}, \\
cx^{-\theta+1}\le g(x)\le dx^{-\theta-2} & \text{if $0<x<\delta$},
\end{cases}
$$
we have $ax^{-\theta}\le dx^2$ for $\delta^{-1}<x<\infty$, which yields that $\theta\ge-2$.
Similarly, we have $\theta\ge-2$ from \eqref{F-2.15} and \eqref{F-2.16}, and also from
\eqref{F-2.15} and \eqref{F-2.17}. This argument does not work when \eqref{F-2.14} and
\eqref{F-2.16} are satisfied. So we take a detour to settle this last case. Since the
function
$$
n^{-1}I_f^\theta(xI_n,yI_n,I_n)=\phi(x,y)^{-\theta}
\quad\mbox{with}\quad\phi(x,y):=f(xy^{-1})y
$$
is jointly convex in $x,y>0$, the Hessian of $\phi(x,y)^{-\theta}$ is non-negative so that
$$
\theta^2\phi^{-2\theta-4}\det\bmatrix(\theta+1)\phi_x^2-\phi\phi_{xx}&
(\theta+1)\phi_x\phi_y-\phi\phi_{xy} \\
(\theta+1)\phi_x\phi_y-\phi\phi_{xy}&
(\theta+1)\phi_y^2-\phi\phi_{yy}
\endbmatrix\ge0.
$$
From \eqref{F-2.13} with $y=1$ we notice that the Hessian of $\phi(x,y)^{-\theta}$ at
$(x,1)$ is
$$
-\theta^2(\theta+1)f(x)^{-2\theta-1}f''(x),
$$
which should be non-negative for all $x>0$. Now assume that $\theta<-1$; then we have
$f''(x)\ge0$ for all $x>0$ so that $f$ is convex on $(0,\infty)$. In the case of
\eqref{F-2.14} with negative $\theta$, we have $f(+0)=0$ and so convexity of $f$ yields
that $f(x)\ge\eps x$ for some $\eps>0$ and for all sufficiently large $x$. From this and
\eqref{F-2.14}, $\eps^{-\theta}x^{-\theta}\le dx^2$ for large $x$, which yields that
$\theta\ge-2$.\qed

\bigskip
It is remarkable that all the convexity/concavity conditions (i)--(vi) except (iv) sit
between (vii) and (viii), and the difference between the last two is only the range
$(0,1]$ or $(0,2]$ of the parameter $\theta$. The equivalence of (i), (v) and (vii) is
also remarkable. It is worth noting that joint concavity (v) of $I_f^\theta(A,B,X)$ in
$A,B\in\bP_n$ occurs only when $\theta\in[-2,0)$ while stronger version (i) or (ii) of
joint convexity does only when $\theta\in(0,2]$.

The following two examples show that implications in (c)--(e) of Theorem \ref{T-2.1} are
almost best possible results.

\begin{example}\label{E-2.6}\rm
Let $f(x):=\sqrt x$ on $(0,\infty)$. According to \cite[Corollary 8.1\,(2)]{Li} the
function
$$
\log I_f^\theta(A,B,X)=\log\Tr X^*A^{-\theta/2}XB^{-\theta/2}
$$
is jointly convex in $(A,B)\in\bP_n\times\bP_n$ for any $\theta\in(0,2]$. Hence (ii)
($\Leftrightarrow$ (iii)) does not imply (vii), and the restriction $\theta\in(0,2]$ from
(ii) is best possible.
\end{example}

\begin{example}\label{E-2.7}\rm
Let $f(x):=x^\alpha$ on $(0,\infty)$, where $\alpha\in\bR$. Recall that the function
$$
I_f^\theta(A,B,X)=\Tr X^*A^{-\alpha\theta}XB^{-(1-\alpha)\theta}
$$
is jointly convex in $(A,B)\in\bP_n\times\bP_n$ if and only if
$(A,B)\in\bP_n\times\bP_n\mapsto A^{-\alpha\theta}\otimes B^{-(1-\alpha)\theta}$ is jointly
convex. According to \cite[p.\ 221, Remark (4)]{An} it is easy to see that this joint
convexity holds if and only if one of the following cases is satisfied:
\begin{itemize}
\item $0\le\alpha\le1$ and
$0<\theta\le\min\bigl\{{1\over\alpha},{1\over1-\alpha}\bigr\}$,
\item $1\le\alpha\le2$ and
$-\min\bigl\{{2\over\alpha},{1\over\alpha-1}\bigr\}\le\theta\le-1$,
\item $-1\le\alpha\le0$ and
$-\min\bigl\{{2\over1-\alpha},{1\over-\alpha}\bigr\}\le\theta\le-1$,
\end{itemize}
where ${1\over0}:=+\infty$ by convention. In particular, (iv) is satisfied when
$\alpha=1/2$ and $\theta\in(0,2]$ or when $\alpha=1$ and $\theta\in[-2,-1]$. Hence the
restriction $\theta\in[-2,-1]\cup(0,2]$ from (iv) is best possible. Also, note that (viii)
does not imply (iv).
\end{example}

\section{Monotonicity of $I_f^\theta(A,B,X)$ and convexity of quasi-entropy}
\setcounter{equation}{0}

A subalgebra of $\bM_n$ means a unital $*$-subalgebra. Given a subalgebra $\cA$ of $\bM_n$
we have the trace-preserving {\it conditional expectation} $E_\cA:\bM_n\to\cA$, which is
determined by
\begin{equation}\label{F-3.1}
\Tr E_\cA(X)Y=\Tr XY,\qquad X\in\bM_n,\ Y\in\cA.
\end{equation}
For our purpose it is convenient to express $E_\cA$ as an average of unitary conjugations.
Let $\cA'$ is the commutant of $\cA$, i.e., $\cA':=\{X\in\bM_n:XY=YX,\,Y\in\cA\}$, and
$\cU(\cA')$ be the set of all unitaries in $\cA'$. Since $\cU(\cA')$ is a compact group,
we have the Haar probability measure on $\cU(\cA')$, which is simply denoted by $dU$. We
then have
\begin{equation}\label{F-3.2}
E_\cA(X)=\int_{\cU(\cA')}UXU^*\,dU,\qquad X\in\bM_n.
\end{equation}
In fact, it is easy to verify that this $E_\cA(X)$ belongs to $\cA$ and satisfies
\eqref{F-3.1}.

For each $n_1,n_2\in\bN$ the $n_1n_2\times n_1n_2$ complex matrix algebra $\bM_{n_1n_2}$
is considered as the tensor product $\bM_{n_1}$ and $\bM_{n_2}$, i.e.,
$\bM_{n_1n_2}=\bM_{n_1}\otimes\bM_{n_2}$. Under this identification a linear map
$\pTr:\bM_{n_1n_2}\to\bM_{n_1}$, called the {\it partial trace}, is determined by
$$
\pTr(X_1\otimes X_2)=\Tr(X_2)X_1,\qquad X_1\in\bM_{n_1},\ X_2\in\bM_{n_2},
$$
which traces out the second factor. Note that $n_2^{-1}\pTr$ is the trace-preserving
conditional expectation from $\bM_{n_1n_2}$ onto the subalgebra
$\cA:=\bM_{n_1}\otimes I_2$, where $I_2$ is the identity of $\bM_{n_2}$.

Given a function $f>0$ on $(0,\infty)$ and $\theta\in\bR\setminus\{0\}$, we consider the
following properties of the function $I_f^\theta$ given in \eqref{F-1.1} concerning
monotonicity under conditional expectations or partial traces:
\begin{itemize}
\item[(I)] for every $n\in\bN$ and any subalgebra $\cA$ of $\bM_n$,
$$
I_f^\theta(E_\cA(A),E_\cA(B),E_\cA(X))\le I_f^\theta(A,B,X),
\qquad A,B\in\bP_n,\ X\in\bM_n,
$$
\item[(IV)] for every $n\in\bN$ and any subalgebra $\cA$ of $\bM_n$,
$$
I_f^\theta(E_\cA(A),E_\cA(B),X)\le I_f^\theta(A,B,X),
\qquad A,B\in\bP_n,\ X\in\cA,
$$
\item[(V)] for every $n\in\bN$ and any subalgebra $\cA$ of $\bM_n$,
$$
I_f^{-\theta}(E_\cA(A),E_\cA(B),X)\ge I_f^{-\theta}(A,B,X),
\qquad A,B\in\bP_n,\ X\in\cA,
$$
\item[(I$'$)] for every $n_1,n_2\in\bN$,
$$
n_2^{\theta-1}I_f^\theta(\pTr A,\pTr B,\pTr X)\le I_f^\theta(A,B,X),
\qquad A,B\in\bP_{n_1n_2},\ X\in\bM_{n_1n_2},
$$
\item[(IV$'$)] for every $n_1,n_2\in\bN$,
$$
n_2^{\theta+1}I_f^\theta(\pTr A,\pTr B,X)\le I_f^\theta(A,B,X\otimes I_2),
\qquad A,B\in\bP_{n_1n_2},\ X\in\bM_{n_1},
$$
\item[(V$'$)] for every $n_1,n_2\in\bN$,
$$
n_2^{1-\theta}I_f^{-\theta}(\pTr A,\pTr B,X)\ge I_f^{-\theta}(A,B,X\otimes I_2),
\qquad A,B\in\bP_{n_1n_2},\ X\in\bM_{n_1}.
$$
\end{itemize}

\begin{thm}\label{T-3.1}
Concerning the above properties and those in Section 2 the following hold:
$$
\mbox{\rm(i)}\Leftrightarrow\mbox{\rm(I)}\Leftrightarrow\mbox{\rm(I$'$)},\quad
\mbox{\rm(iv)}\Leftrightarrow\mbox{\rm(IV)}\Leftrightarrow\mbox{\rm(IV$'$)},\quad
\mbox{\rm(v)}\Leftrightarrow\mbox{\rm(V)}\Leftrightarrow\mbox{\rm(V$'$)}.
$$
\end{thm}

\begin{proof}
We will prove only the equivalence of (i), (I) and (I$'$) since other statements can
similarly be proved.

\medskip
(i) $\Rightarrow$ (I).\enspace
Thanks to \eqref{F-3.2} this is seen as follows:
\begin{align*}
&I_f^\theta(E_\cA(A),E_\cA(B),E_\cA(X)) \\
&\qquad=I_f^\theta\biggl(\int_{\cU(\cA')}UAU^*\,dU,\int_{\cU(\cA')}UBU^*\,dU,
\int_{\cU(\cA')}UXU^*\,dU\biggr) \\
&\qquad\le\int_{\cU(\cA')}I_f^\theta(UAU^*,UBU^*,UXU^*)\,dU
=I_f^\theta(A,B,X).
\end{align*}
Here, it is obvious that $I_f^\theta(UAU^*,UBU^*,UXU^*)$ is continuous (hence integrable)
in $U\in\cU(\cA')$.

\medskip
(I) $\Rightarrow$ (I$'$) is immediate since $X\mapsto{n_2^{-1}}\pTr X\otimes I_2$ is the
conditional expectation from $\bM_{n_1n_2}$ onto $\bM_1\otimes I_2$ and
$$
I_f^\theta(n_2^{-1}\pTr A\otimes I_2,n_2^{-1}\pTr B\otimes I_2,n_2^{-1}\pTr X\otimes I_2)
=n_2^{\theta-1}I_f^\theta(\pTr A,\pTr B,\pTr X).
$$

\medskip
(I$'$) $\Rightarrow$ (i).\enspace
For $A_1,A_2,B_1,B_2\in\bP_n$ and $X_1,X_2\in\bM_n$ set
$$
A:=\bmatrix A_1&0\\0&A_2\endbmatrix,\quad
B:=\bmatrix B_1&0\\0&B_2\endbmatrix,\quad
X:=\bmatrix X_1&0\\0&X_2\endbmatrix\quad\mbox{in $\bM_n\otimes\bM_2$}.
$$
Since $\pTr A=A_1+A_2$, $\pTr B=B_1+B_2$ and $\pTr X=X_1+X_2$, (I$'$) for $n_2=2$ implies
that
$$
2^{\theta-1}I_f^\theta(A_1+A_2,B_1+B_2,X_1+X_2)\le I_f^\theta(A,B,X).
$$
The above left-hand side is
$$
2I_f^\theta\biggl({A_1+A_2\over2},{B_1+B_2\over2},{X_1+X_2\over2}\biggr)
$$
while the right-hand side is $I_f^\theta(A_1,B_1,X_1)+I_f^\theta(A_2,B_2,X_2)$. Hence (i)
follows.
\end{proof}

The proof of the above (i) $\Rightarrow$ (I) is similar to those in \cite{CL,JR} where the
method of representing a partial trace as an average of unitary conjugations was used.

The next corollary is immediate from Theorems \ref{T-2.1} and \ref{T-3.1}.

\begin{cor}
For every operator monotone function $f>0$ on $(0,\infty)$ and for every $\theta\in(0,1]$,
all the properties {\rm(i)--(vi)}, {\rm(I)}, {\rm(IV)}, {\rm(V)}, {\rm(I$'$)}, {\rm(IV$'$)}
and {\rm(V$'$)} hold.
\end{cor}

When $\theta<0$, the function
$$
I_f^\theta(A,B,X)=\<X,(f(L_AR_B^{-1})R_B)^{-\theta}X\>_\HS,
\qquad A,B\in\bP_n,\ X\in\bM_n,
$$
is well defined when $f$ is a real (not necessarily positive) function on $(0,\infty)$.
In particular, when $\theta=-1$, $I_f^{-1}(A,B,X)$ is the quasi-entropy $S_f^X(A\|B)$ (see
Section 1.2).

\begin{thm}\label{T-3.3}
Let $f$ be a real function on $(0,\infty)$. Then the following conditions are equivalent:
\begin{itemize}
\item[\rm(c1)] $(A,B)\in\bP_n\times\bP_n\mapsto S_f^X(A\|B)$ is jointly convex for any
fixed $X\in\bM_n$ and for every $n\in\bN$;
\item[\rm(c2)] for every $n\in\bN$ and any subalgebra $\cA$ of $\bM_n$,
$$
S_f^X(E_\cA(A)\|E_\cA(B))\le S_f^X(A\|B),
\qquad A,B\in\bP_n,\ X\in\cA;
$$
\item[\rm(c3)] for every $n_1,n_2\in\bN$,
$$
S_f^X(\pTr A\|\pTr B)\le S_f^{X\otimes I_2}(A\|B),
\qquad A,B\in\bP_{n_1n_2},\ X\in\bM_{n_1};
$$
\item[\rm(c4)] $f$ is operator convex on $(0,\infty)$;
\item[\rm(c5)] $f(x^{-1})x$ is operator convex on $(0,\infty)$.
\end{itemize}
\end{thm}

\begin{proof}
Conditions (c1), (c2) and (c3) are nothing but (iv), (IV) and (IV$'$), respectively,
with $\theta=-1$, whose equivalence is in Theorem \ref{T-3.1}. (c4) $\Rightarrow$ (c1)
and (c4) $\Rightarrow$ (c2) were given in \cite{Pe1} (see also \cite{Eff}). In the proof
of Theorem \ref{T-2.1}\,(e) we saw that (iv) implies that $f(x)^{-\theta}$ and
$(f(x^{-1})x)^{-\theta}$ are operator convex on $(0,\infty)$. In particular, when
$\theta=-1$, this shows that (c1) implies (c4) and (c5). (Note that under $\theta=-1$ we
did not use the positivity assumption for $f$ in this part of the proof and also in the
proof of Theorem \ref{T-3.1}.) Hence we have (c4) $\Rightarrow$ (c5), which gives also
(c5) $\Rightarrow$ (c4).
\end{proof}

The equivalence of (c4) and (c5) for a real function $f$ seems new.

\begin{remark}\label{R-3.4}\rm
Conditions (c1)--(c3) with restriction of $A,B$ to density matrices are also equivalent
to the conditions in Theorem \ref{T-3.3}. Indeed, write (c1$'$)--(c3$'$) for (c1)--(c3)
with this restriction. Then (c1) $\Leftrightarrow$ (c1$'$) is
(iv) $\Leftrightarrow$ (iv$'$) of Theorem \ref{T-2.1}\,(a) with $\theta=-1$. It is
immediate to check that Theorem \ref{T-3.1} holds with restriction of $A,B$ to density
matrices. This means that (c1$'$)--(c3$'$) are equivalent.
\end{remark}

\begin{thm}\label{T-3.5}
For any real function $f$ on $(0,\infty)$ that is not identically zero,
$S_X^f(\rho\|\sigma)$ is not jointly convex in
$(\rho,\sigma,X)\in\cD_n\times\cD_n\times\bM_n$ for some $n\in\bN$.
\end{thm}

\begin{proof}
Suppose that $S_X^f(\rho\|\sigma)$ is jointly convex in
$(\rho,\sigma,X)\in\cD_n\times\cD_n\times\bM_n$ for every $n$. By Theorem \ref{T-2.1}\,(a)
(here the positivity assumption for $f$ is irrelevant), so is $S_X^f(A\|B)$ in
$(A,B,X)\in\bP_n\times\bP_n\times\bM_n$ for every $n$. Then as in the proof of
(i) $\Rightarrow$ (vii) of Theorem \ref{T-2.1}\,(c), it follows that $\<\xi,f(A)\xi\>$ and
$\<\xi,f(A^{-1})A\xi\>$ are jointly convex in $(A,\xi)\in B(\cH)^{++}\times\cH$. Hence by
\cite[Remark 3.5]{AH} both $f$ and $f(x^{-1})x$ are non-negative and operator monotone
decreasing. But this is impossible unless $f$ is identically zero.
\end{proof}

\begin{remark}\rm
When $f$ is a real function on $[0,\infty)$, the definition of the quasi-entropy
$S_f^X(A\|B)$ for general $A,B\in\bM_n^+$ and $X\in\bM_n$ (\cite{Ko,Pe2}) is
$$
S_f^X(A\|B):=\<XB^{1/2},f(L_AR_{B^{-1}})XB^{1/2}\>_\HS,
$$
where  $B^{-1}$ is defined in the sense of generalized inverse. (Recall that the relative
modular operator for $\Tr(A\,\cdot)$ and $\Tr(B\,\cdot)$ coincides with $L_AR_{B^{-1}}$
in this sense.) Here, assume that $\omega(f):=\lim_{x\to\infty}f(x)/x$ exists in
$[-\infty,\infty]$. Then we notice (as in \cite[Proposition 2.2]{HMPB} in the case $X=I$)
that
\begin{equation}\label{F-3.3}
\lim_{\eps\searrow0}S_f^X(A\|B+\eps I)=S_f^X(A\|B)+\omega(f)\Tr X^*AX(I-B^0),
\end{equation}
where $B^0$ means the support projection of $B$. We write $\tilde S_f^X(A\|B)$ for the
above identical expressions. Obviously, $\tilde S_f^X(A\|B)=S_f^X(A\|B)$ if $B\in\bP_n$.
Now we show that if $f$ is an operator convex function on $[0,\infty)$, then
$\tilde S_f^X(A\|B)$ is jointly convex in $(A,B)\in\bM_n^+\times\bM_n^+$ for every
$X\in\bM_n$. Indeed, it is clear that $\omega(f)$ exists in $(-\infty,\infty]$ for convex
$f$. From definition in the left-hand side of \eqref{F-3.3}, it suffices to prove the
joint convexity on $\bM_n^+\times\bP_n$. But this is immediate from the joint convexity
on $\bP_n\times\bP_n$ (Theorem \ref{T-3.3}) and  the continuity of
$A\in\bM_n^+\mapsto S_f^X(A\|B)$ with fixed $B\in\bP_n$. Furthermore, it is easy to see
that
$$
(A,B)\in\bM_n^+\times\bM_n^+\mapsto\Tr X^*AX(I-B^0)
$$
is jointly convex. Thus, when $f$ is a non-negative opreator monotone function on
$[0,\infty)$ (hence $-f$ is operator convex), we notice that
$$
S_f^X(A\|B)=\tilde S_f^X(A\|B)-\omega(f)\Tr X^*AX(I-Q)
$$
is jointly concave in $(A,B)\in\bM_n^+\times\bM_n^*$. This is the joint concavity results
in \cite{Ko,Pe1} though restricted to matrices. The above argument also clarifies why
the assumption of $B\in\bP_n$ is essential for the joint convexity result in
\cite{Pe2,HMPB} when $f$ is operator convex. In this way, joint convexity of
$\tilde S_f^X(A\|B)$ covers all the known joint concavity/convexity results for
$S_f^X(A\|B)$ in \cite{Ko,Pe1,Pe2} (also \cite{HMPB} where $\tilde S_f^X(A\|B)$ with $X=I$
was denoted by $S_f(A\|B)$).
\end{remark}

\section{Convexity of skew information and quantum $\chi^2$-divergence}
\setcounter{equation}{0}

Given a general function $f>0$ on $(0,\infty)$ we define the (unbounded version of)
{\it $f$-skew information}
$$
\cI_A^f(X):=\<i[A,X],(J_A^f)^{-1}(i[A,X])\>_\HS=I_f^1(A,A,i[A,X])
$$
and the {\it quantum $f$-$\chi^2$-divergence}
$$
\chi_f^2(A,B):=\<A-B,(J_B^f)^{-1}(A-B)\>_\HS=I_f^1(B,B,A-B)
$$
for each $A,B\in\bP_n$ and $X\in\bM_n$. When $A=\rho$, $B=\sigma$ with
$\rho,\sigma\in\cD_n$ and $f$ is an operator monotone function, $\cI_\rho(X)$ is the
unbounded version of the metric adjusted skew information \eqref{F-1.7} and
$\chi_f^2(\rho,\sigma)$ is the quantum $\chi^2$-divergence \eqref{F-1.8}.

We define the {\it harmonic symmetrization} $f^\sym$ of $f$ to be $f_{-1,\sym}$ given in
\eqref{F-2.1} with $\theta=1$, i.e., $(-1)$-power symmetrization of $f$. In this section
we show the next theorem, which extend convexity results in \cite{Ha2,CH,TKRWV,Ha3} (see
Section 1.4) into a combined characterization theorem.

\begin{thm}\label{T-4.1}
Let $f>0$ be a function on $(0,\infty)$. Then the following conditions are equivalent:
\begin{itemize}
\item[\rm(d1)] $A\in\bP_n\mapsto\cI_A^f(K)$ is convex for any fixed $K\in\bH_n$ and for
every $n\in\bN$;
\item[\rm(d2)] $\rho\in\cD_n\mapsto\cI_\rho^f(K)$ is convex for any fixed $K\in\bH_n$ and for
every $n\in\bN$;
\item[\rm(d3)] $(A,B)\in\bP_n\times\bP_n\mapsto\chi_f^2(A,B)$ is jointly convex
for every $n\in\bN$;
\item[\rm(d4)] $(\rho,\sigma)\in\cD_n\times\cD_n\mapsto\chi_f^2(\rho,\sigma)$ is
jointly convex for every $n\in\bN$;
\item[\rm(d5)] $(x-1)^2/f^\sym(x)$ is operator convex on $(0,\infty)$;
\item[\rm(d6)] $f^\sym$ is operator monotone on $(0,\infty)$.
\end{itemize}
\end{thm}

We first give the following lemma. The equivalence between (1) and (4) will be used in the
proof of the theorem. Other conditions (2) and (3) are stated for the convenience of the
proof and also for the completeness of statements.

\begin{lemma}\label{L-4.2}
Let $f>0$ be a function on $(0,\infty)$. Then the following conditions are equivalent:
\begin{itemize}
\item[\rm(1)] $f$ is operator monotone on $(0,\infty)$;
\item[\rm(2)] $(x-1)/f(x)$ is operator monotone on $(0,\infty)$;
\item[\rm(3)] $(x-1)f(x)$ is operator convex on $(0,\infty)$;
\item[\rm(4)] $(x-1)^2/f(x)$ is operator convex on $(0,\infty)$.
\end{itemize}
\end{lemma}

\begin{proof}
(1) $\Rightarrow$ (2).\enspace
L\"owner's theorem tells that $f$ is analytically continued to a Pick function $f(z)$
defined on $\bC^+$, so $f(z)$ maps $\bC^+$ into $\bC^+\cup\,\bR$ and moreover
$\arg f(z)\le\arg z$ (here argument is taken in $[0,\pi)$) for every $z\in\bC^+$. (The
last fact on $\arg f(z)$ might not be familiar but it is easily verified by using the
integral representation of $f$.) Then $g(z):=(z-1)/f(z)$ is well defined as an analytic
function on $\bC^+$. (If $f$ is a constant  $\alpha>0$, then $f(z)$ is a constant $\alpha$.
If $f$ is not constant, then $f(\bC^+)\subset\bC^+$ so that $f(z)$ is not zero for any
$z\in\bC^+$.) When $z=re^{i\eta}$ with $r>0$ and $\eta\in(0,\pi)$, we notice that
$z-1=r_1e^{i\eta_1}$ with $r_1>0$ and $\eta<\eta_1<\pi$, and that $f(z)=r_2e^{i\eta_2}$
with $r_2>0$ and $0\le\eta_2\le\eta$. Therefore,
$$
\Im g(z)={\Im\{(z-1)\overline{f(z)}\}\over|f(z)|^2}
$$
and
$$
\Im\{(z-1)\overline{f(z)}\}=\Im\{r_1e^{i\eta_1}\cdot r_2e^{-i\eta_2}\}
=r_1r_2\,\Im e^{i(\eta_1-\eta_2)}>0,
$$
since $0<\eta_1-\eta_2<\pi$. L\"owner's theorem implies that $g$ is operator monotone on
$(0,\infty)$.

\medskip
(2) $\Rightarrow$ (1).\enspace
Assume that $g(x):=(x-1)/f(x)$ is operator monotone on $(0,\infty)$. Then $g(1)=0$, and
by \cite[Theorem 1.9]{FHR} there exist a $\gamma>0$ and a positive measure $\mu$ on
$[0,\infty)$ such that
$$
\int_{[0,\infty)}{1\over(1+\lambda)^2}\,d\mu(\lambda)<+\infty
$$
and
$$
g(x)=\gamma(x-1)+\int_{[0,\infty)}{x-1\over(x+\lambda)(1+\lambda)}\,d\mu(\lambda),
\qquad x\in(0,\infty).
$$
Therefore, we have
$$
{1\over f(x)}=\gamma+\int_{[0,\infty)}{1\over(x+\lambda)(1+\lambda)}\,d\mu(\lambda)
=\gamma+\int_{[0,\infty)}{1+\lambda\over x+\lambda}\,d\nu(\lambda),
$$
where $\nu$ is a finite positive measure on $[0,\infty)$ given by
$d\nu(\lambda):=d\mu(\lambda)/(1+\lambda)^2$. The above integral expression shows
(see \cite{Ha1}, \cite[Theorem 3.1]{AH}) that $1/f(x)$ is operator monotone decreasing on
$(0,\infty)$, so $f$ is operator monotone on $(0,\infty)$.

(1) $\Leftrightarrow$ (3) and (2) $\Leftrightarrow$ (4) are seen from \cite[Lemma 2.1]{Uc}
(also \cite[Corollary 2.7.8]{Hi}).
\end{proof}

We note that (1) $\Rightarrow$ (4) was shown in \cite{CH} in a different (and more
tractable) method, however the above proof has a merit to show the equivalence of (1) and
(4). As remarked in \cite{CH}, when $f(x)=(x-1)^2$ on $(0,\infty)$, we have
$(x-1)^2/f(x)\equiv1$ but $f$ is not operator monotone. Hence the assumption $f>0$ cannot
be relaxed to $f\ge0$ for (4) $\Rightarrow$ (1).

\bigskip\noindent
{\it Proof of Theorem \ref{T-4.1}.}\enspace
The proof of (d1) $\Leftrightarrow$ (d2) is similar to the last part of the proof of
Theorem \ref{T-2.1}\,(a), and (d5) $\Leftrightarrow$ (d6) follows from Lemma \ref{L-4.2}.
To prove that (d1) $\Leftrightarrow$ (d5), define $\tilde f(x):=f(x^{-1})x$ and
\begin{equation}\label{F-4.1}
h(x):={(x-1)^2\over f(x)},\quad\tilde h(x):=h(x^{-1})x,
\quad h_\sym(x):={h(x)+\tilde h(x)\over2},\qquad x>0,
\end{equation}
so that
$$
\tilde h(x)={(x-1)^2\over\tilde f(x)},
\quad h_\sym(x)={(x-1)^2\over f^\sym(x)}.
$$
As mentioned in \cite{Ha2,CH}, for every $A\in\bP_n$ and $X\in\bM_n$ we notice that
\begin{equation}\label{F-4.2}
\cI_A^f(X)=\<X,J_A^h(X)\>_\HS,\quad
\cI_A^{\tilde f}(X)=\<X,J_A^{\tilde h}(X)\>_\HS.
\end{equation}
Furthermore, similarly to \eqref{F-2.4} we have
\begin{equation}\label{F-4.3}
\cI_A^f(X)=\cI_A^{\tilde f}(X^*),\qquad X\in\bM_n,
\end{equation}
so that for every $K\in\bH_n$,
\begin{equation}\label{F-4.4}
\cI_A^f(K)=\cI_A^{\tilde f}(K)=\bigl\<K,J_A^{h_\sym}K\bigr\>_\HS=S_{h_\sym}^K(A\|A).
\end{equation}
For every $A,B\in\bP_n$ and $X\in\bM_n$ set
$$
\tilde A:=\bmatrix A&0\\0&B\endbmatrix\in\bP_{2n},\quad
\tilde K:=\bmatrix0&X\\X^*&0\endbmatrix\in\bH_{2n}.
$$
Thanks to \eqref{F-4.4}, \eqref{F-2.2} and \eqref{F-2.4} we have
\begin{align}
\cI_{\tilde A}^f(\tilde K)
&=\bigl\<\tilde K,J_{\tilde A}^{h_\sym}\tilde K\bigr\>_\HS
=\<X,J_{A,B}^{h_\sym}X\>_\HS+\<X^*,J_{B,A}^{h_\sym}X^*\>_\HS \nonumber\\
&=2\<X,J_{A,B}^{h_\sym}X\>_\HS=2S_{h_\sym}^X(A\|B). \label{F-4.5}
\end{align}
From \eqref{F-4.4} and \eqref{F-4.5} one can see that (d1) holds if and only if
condition (c1) of Theorem \ref{T-3.3} holds for $h_\sym$ in place of $f$. The latter
condition is equivalent to (d5) by Theorem \ref{T-3.3}.

Now we turn to conditions (d3) and (d4). (d3) $\Rightarrow$ (d4) is trivial. Since
$\chi_f^2(A,B)=\chi_{f^\sym}^2(A,B)$ for all $A,B\in\bP_n$ due to \eqref{F-2.5},
(d6) $\Rightarrow$ (d3) follows from (vii) $\Rightarrow$ (i) of Theorem 2.1\,(c).
It remains to prove that (d4) $\Rightarrow$ (d6). Assume (d4) so that
$(\rho,\sigma)\in\cD_n\times\cD_n\mapsto\chi_{f^\sym}^2(\rho,\sigma)$ is jointly convex.
Let $\sigma_1,\sigma_2\in\cD_n$, $H_1,H_2\in\bH_n^0$ and $0<\lambda<1$. Choose an $\eps>0$
such that $\sigma_i+\eps H_i$, $i=1,2$, are positive definite. Set
$\rho_i:=\sigma_i+\eps H_i\in\cD_n$, $i=1,2$. Then we have
\begin{align*}
&I_{f^\sym}^1(\lambda\sigma_1+(1-\lambda)\sigma_2,\lambda\sigma_1+(1-\lambda)\sigma_2,
\lambda H_1+(1-\lambda)H_2) \\
&\qquad=\eps^{-2}\chi_{f^\sym}^2(\lambda\rho_1+(1-\lambda)\rho_2,
\lambda\sigma_1+(1-\lambda)\sigma_2) \\
&\qquad\le\lambda\eps^{-2}\chi_{f^\sym}^2(\rho_1,\sigma_1)
+(1-\lambda)\eps^{-2}\chi_{f^\sym}^2(\rho_2,\sigma_2) \\
&\qquad=\lambda I_{f^\sym}^1(\sigma_1,\sigma_1,H_1)
+(1-\lambda)I_{f^\sym}^1(\sigma_2,\sigma_2,H_2),
\end{align*}
which means that $(\sigma,H)\in\cD_n\times\bH_n^0\mapsto I_{f^\sym}^1(\sigma,\sigma,H)$ is
jointly convex, that is, condition (i$''$) in Section 2 holds with $\theta=1$ for $f^\sym$.
Hence (d6) follows by (b) and (c) of Theorem \ref{T-2.1}.\qed

\bigskip
When (d1) is replaced with the stronger condition that $\cI_A^f(X)$ is convex in
$A\in\bP_n$ for any fixed $X\in\bM_n$, we have the next theorem. The proof is similar to
(indeed, a bit simpler than) that of Theorem \ref{T-4.1}, so we omit it.

\begin{thm}\label{T-4.3}
Let $f>0$ be a function on $(0,\infty)$. Then the following conditions are equivalent:
\begin{itemize}
\item[\rm(d1$'$)] $A\in\bP_n\mapsto\cI_A^f(X)$ is convex for any fixed $X\in\bM_n$ and for
every $n\in\bN$;
\item[\rm(d2$'$)] $\rho\in\cD_n\mapsto\cI_\rho^f(X)$ is convex for any fixed $X\in\bM_n$
and for every $n\in\bN$;
\item[\rm(d5$'$)] $(x-1)^2/f(x)$ is operator convex on $(0,\infty)$;
\item[\rm(d6$'$)] $f$ is operator monotone on $(0,\infty)$.
\end{itemize}
\end{thm}

Although it is obvious from \eqref{F-4.2} that $\cI_A^f(X)$ is convex in $X\in\bM_n$ for
any fixed $A\in\bP_n$, the function $\cI_\rho^f(K)$ cannot be jointly convex in
$(\rho,K)\in\cD_n\times\bH_n$.

\begin{thm}
For any function $f>0$ on $(0,\infty)$, the function $\cI_\rho^f(K)$ is not jointly
convex in $(\rho,K)\in\cD_n\times\bH_n$ for some $n\in\bN$.
\end{thm}

\begin{proof}
Suppose that $I_A^f(K)$ is jointly convex in $(\rho,K)\in\cD_n\times\bH_n$ for every
$n\in\bN$. Then from \eqref{F-4.5} it must follow that $S_{h_\sym}^X(\rho\|\sigma)$ is
jointly convex in $(\rho,\sigma,X)\in\cD_n\times\cD_n\times\bM_n$ for every $n\in\bN$,
where $h_\sym$ is given in \eqref{F-4.1}. However this is impossible by Theorem
\ref{T-3.5}.
\end{proof}

\end{document}